\documentclass[letterpaper, 10 pt, conference]{ieeeconf}  

\IEEEoverridecommandlockouts                              

\usepackage{graphics} 
\usepackage{subcaption}
\usepackage{graphicx}
\usepackage{epsfig} 
\usepackage{float}
\usepackage{mathptmx} 
\usepackage{times} 
\usepackage{amsmath} 
\usepackage{amssymb}  
\newtheorem{theorem}{Theorem}

\title{\LARGE \bf
Trajectory Tracking Using Motion Primitives for the Purcell’'s Swimmer
}

\author{Sudin Kadam$^{1}$, Kedar Joshi$^{2}$, Naman Gupta$^{2}$, Pulkit Katdare$^{2}$ and Ravi Banavar$^{3}$
\thanks{$^{1}$ Sudin Kadam is with Systems and Control Engineering, Indian Institute of Technology Bombay, Powai, Mumbai, India, 400076
        {\tt\small sudin@sc.iitb.ac.in}}%
\thanks{$^{2}$ Kedar Joshi, Naman Gupta and Pulkit Katdare are with Department of Mechanical Engineering, Indian Institute of Technology Bombay, Powai, Mumbai, India, 400076
        {\tt\small kedarjoshi@iitb.ac.in, namangupta@iitb.ac.in, pulkitkatdare@iitb.ac.in}}%
\thanks{$^{3}$ Ravi Banavar is a professor at Systems and Control Engg, IIT Bombay (currently, Visiting Professor, Department of Electrical Engg.,  IIT-Gandhinagar.)
        {\tt\small banavar@iitb.ac.in}}%
}

\begin{document}

\maketitle
\thispagestyle{empty}
\pagestyle{empty}

\begin{abstract}
Locomotion at low Reynolds numbers is a topic of growing interest, spurred by its various engineering and medical applications. This paper presents a novel prototype and a locomotion algorithm for the 3-link planar Purcell's  swimmer based on Lie algebraic notions. The kinematic model based on Cox theory of the prototype swimmer is a driftless control-affine system. Using the existing strong controllability and related results, the existence of motion primitives is initially shown. The Lie algebra of the control vector fields is then used to synthesize control profiles to generate motions along the basis of the Lie algebra associated with the structure group of the system. An open loop control system with vision-based  positioning is successfully implemented which allows tracking any given continuous trajectory of the position and orientation of the swimmer's base link. Alongside, the paper also provides a theoretical interpretation of the symmetry arguments presented in the existing literature to generate the control profiles of the swimmer.
\end{abstract}

\section{INTRODUCTION}
Locomotion relates to a variety of movements resulting in transportation from one place to another, and is crucial to existential requirements of microbial and animal life. The primary objectives in the studies of locomotion problems are to understand the mechanics of locomotion of existing biological systems and to devise mobile robotic systems to mimick similar biological systems. The type of interaction with the environment to achieve the motion forms the key to analysis of such systems. Our work focuses on the $3$-link planar Purcell's swimmer, a mechanism proposed by E. M. Purcell in his famous lecture on life at low Reynolds numbers \cite{5}.

As opposed to inertia-dominant systems, microbial motion occurs in a fluid medium with very low Reynolds number conditions, which is the ratio of the inertial to viscous forces acting on the swimmer's body.  The Reynolds number in microbial motion regimes is of the order of $10^{-3}$. To get a relative sense of the numbers, the Reynolds number for a man swimming in water is of the order of $10^4$, whereas that for a Goldfish swimming in water is of the order of $10^2$ \cite{46}, \cite{63}. Clearly, at low Reynolds numbers, viscous forces strongly dominate the motion, which is contrary to that of locomotion of larger animals which have prominent inertial effects. Swimming in this regime has seen a lot of interest, not only because of curiosity in mechanics of motion used by microbes, but also due to many engineering applications such as targeted drug delivery, micromachining, non-invasive surgery \cite{6}, \cite{7}, \cite{8}, \cite{33}, \cite{48} that rely on the knowledge in this domain.

The slender body theory at low Reynolds number helps in obtaining the equations of motions which fit well in the geometric mechanics and control theoretic framework. For a large class of locomotion systems, including underwater vehicles, spacecraft with rotors and wheeled or legged robots, it is possible to model the motion using the mathematical structure of a connection on a principal bundle, see \cite{10}, \cite{35} for examples. We use such a geometric approach to model the Purcell's swimmer's motion, which has also been followed in \cite{10}, \cite{35}. In addition to the kinematic modelling, the problems on controllability of the Purcell's swimmer and its variants has also been widely studied in the literature \cite{11}, \cite{12}, \cite{44}, \cite{38}.

The literature has also explained approaches to problem of motion planning of the Purcell's swimmer. The concept of body-velocity integral is introduced in \cite{42}, in which the velocities expressed in a body-fixed frame are integrated in order to find the net body motion for a given gait. The original paper by Purcell \cite{5} has used intuitive arguments to explain square gait by making use of structural symmetries in the swimmer model for generating motions. Based on the concept of principal connection resulting out of the non-holonomic constraints, a few other robotic locomotion problems have also been analyzed, like classical examples of snakeboard and roller-racer. Work in \cite{35}. \cite{54}, \cite{56} present optimal gait design strategy for such class of swimmers, whereas \cite{55} explains control using curvature of connection forms of the system.  In one of the recent works, \cite{62} identifies the symmetries in the Purcell's model using which certain sequence of base curves are identified which generate a net displacement in one of the group directions. Although this method works for the 3-link Purcell's swimmer, the shortcomings of this approach are that one has to explicitly identify such symmetries and there does not seem to be a standard approach to identify the gaits which lead to a particular group motion. 

The approach presented in this paper bases itself on the notions of controllability, existence of motion primitives and properties of flow of vector fields under lie bracket operations, and is not specific to a particular system model. The problem of finding the group motion is simply treated as finding an element in the Lie algebra of the underlying control vector fields. Moreover, this way of constructing motion allows us to arbitrarily select the trajectories in the group variable, given at least the weak controllability property of the underlying system \cite{10}. Such an approach is based on the theory of motion primitives, presented in \cite{66} for solving problems of point-to-point reconfiguration of satellite and underwater vehicle control systems. The idea is also used in \cite{67} for human-inspired bipedal robotic locomotion.

\subsection{Organization of the paper:}
In the following section we discuss the configuration space of the swimmer and briefly go over its kinematic model based on resistive force theory. In section 3, we explain the experimental setup of the swimmer's hardware prototype, the arrangement of vision system and the approach to obtain the viscous drag coefficients of the swimmer. Section 4 describes the theory of motion using primitives for the driftless systems. We prove the existence of the primitive controls followed by their synthesis. Finally, in the last section we present the experimental results and explain the analogy between the gaits based on symmetry arguments shown in \cite{62}.

\section{Kinematic model of the Purcell's swimmer}
The Purcell's swimmer is a 3 link mechanism moving in a fluid at low Reynolds number. Each link of the swimmer is modelled as a rigid slender body of length $2L$ and radius $b$. The outer 2 links are actuated through rotary joints with the base link. We represent the orientation of the outer links with respect to the base link through shape variables $\alpha_1, \alpha_2$, each of which evolves on a circle $\mathbb{S}^1$. Thus, the shape space of the mechanism is parametrized by the two joint angles $(\alpha_1, \alpha_2) \in \mathbb{S}^1\times\mathbb{S}^1$. The macro-position of the three-link system in Fig. \ref{original_purcell01} is defined by location of the midpoint of the base link and its orientation with respect to the inertial reference frame. This is represented by $g$, which belongs to the Special Euclidean group $SE(2)$, which is the Lie group in our example parametrized by $(x,\: y,\: \theta)$. For a Lie group $G$, its associated Lie algebra is denoted by $\mathfrak{g}$. We denote by $\xi \in \mathfrak{g}$ the body frame velocity of the base link such that $\xi=[\xi_x, \xi_y, \xi_{\theta}]^T$, with $\xi_x, \xi_y$ being the translational components of velocity of the base link and $\xi_{\theta}$ is its rotational component. Hence, the configuration space is 
\begin{equation}
Q\:=\:SE(2) \times\mathbb{S}^1\times\mathbb{S}^1
\end{equation}

\begin{figure}[!htb]
\centering
\includegraphics[scale=.4]{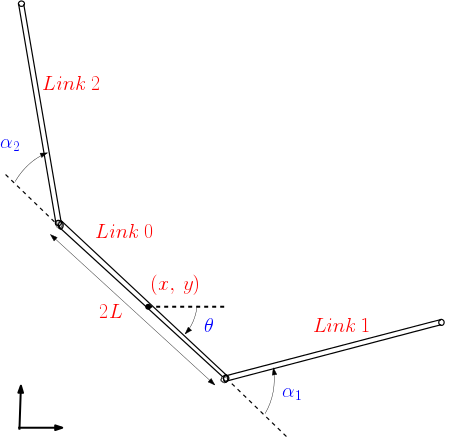}
\caption{Purcell's planar swimmer \cite{42}}
\label{original_purcell01}
\end{figure}

The condition of being at low Reynold's number and slenderness of the links in mechanism form the basis for modeling of the fluid forces acting on the systems in this regime. We use these to obtain the forces acting on an infintesimal element of length $\ell$ of a link moving in a fluid with viscosity $\mu$ corresponding  to its lateral and longitudinal direction in terms of its velocities as follows \cite{41}, \cite{57}
\begin{equation}\nonumber
F_{long} = \frac{2\pi\mu\ell \xi_{long}}{\ln(\ell/b)} = k_{long}\xi_{long},\:\:\: F_{lat} = \frac{4\pi\mu\ell \xi_{lat}}{\ln(\ell/b)} = k_{lat}\xi_{lat}
\end{equation}
where $\xi_{lat}$ and $\xi_{long}$ are the local lateral (perpendicular to the link length direction) and longitudinal (along the link length) velocity components. The total force can then be obtained by integrating over the entire link. By virtue of the resistive force theory \cite{42}, and the assumption of swimmer being massless we get the resulting equations of motion in a purely kinematic form in quation ($\ref{pure_kinematic}$). See \cite{24} for details on purely kinematic systems.
\begin{equation}\label{pure_kinematic}
\xi = -A(r)\dot{r}
\end{equation}
where $A(r)$ is the local connection form defined at each $r \in M$. We recall that for a shape space $M$, which is a manifold, its tangent space at a point $r \in M$ is denoted by $T_rM$, and the shape velocity $\dot{r} = (\dot{\alpha}_1,\dot{\alpha}_2) \in T_rM$. The local connection form is thus defined as $A(r) : T_rM \longrightarrow \mathfrak{g}$. For the Purcell's swimmer local connection form $A(r)$ is a $3 \times 3$ matrix which appears in the form of $\omega_1^{-1}\omega_2$. The matrices $\omega_1$ and $\omega_2$ are of size $3 \times 3$ and $3 \times 2$, respectively, and they depend on the lengths of the limbs, viscous drag coefficient $k$ and the shape of mechanism $(\alpha_1, \alpha_2)$. We refer to \cite{42} for explicit form of $\omega_1$ and $\omega_2$. The connection form and the other notions mentioned here have roots in geometric mechanics, see \cite{24}, \cite{35} for details.

\section{Experimental Setup}
The system model in equation \ref{pure_kinematic} forms the basis of our motion control analysis. This section looks into the experimental setup which is used to validate the theoretical model of the swimmer and the motion control algorithm. Although the 3-link swimmer used for the experimentation is a macro-scale model with links as flat plates and not of microbial scale, the high viscosity of fluid and very slow actuation speeds gives us the Reynolds number values low enough to have negligible inertial effects.

\begin{figure}[h!]
\centering
\includegraphics[width=0.4\textwidth]{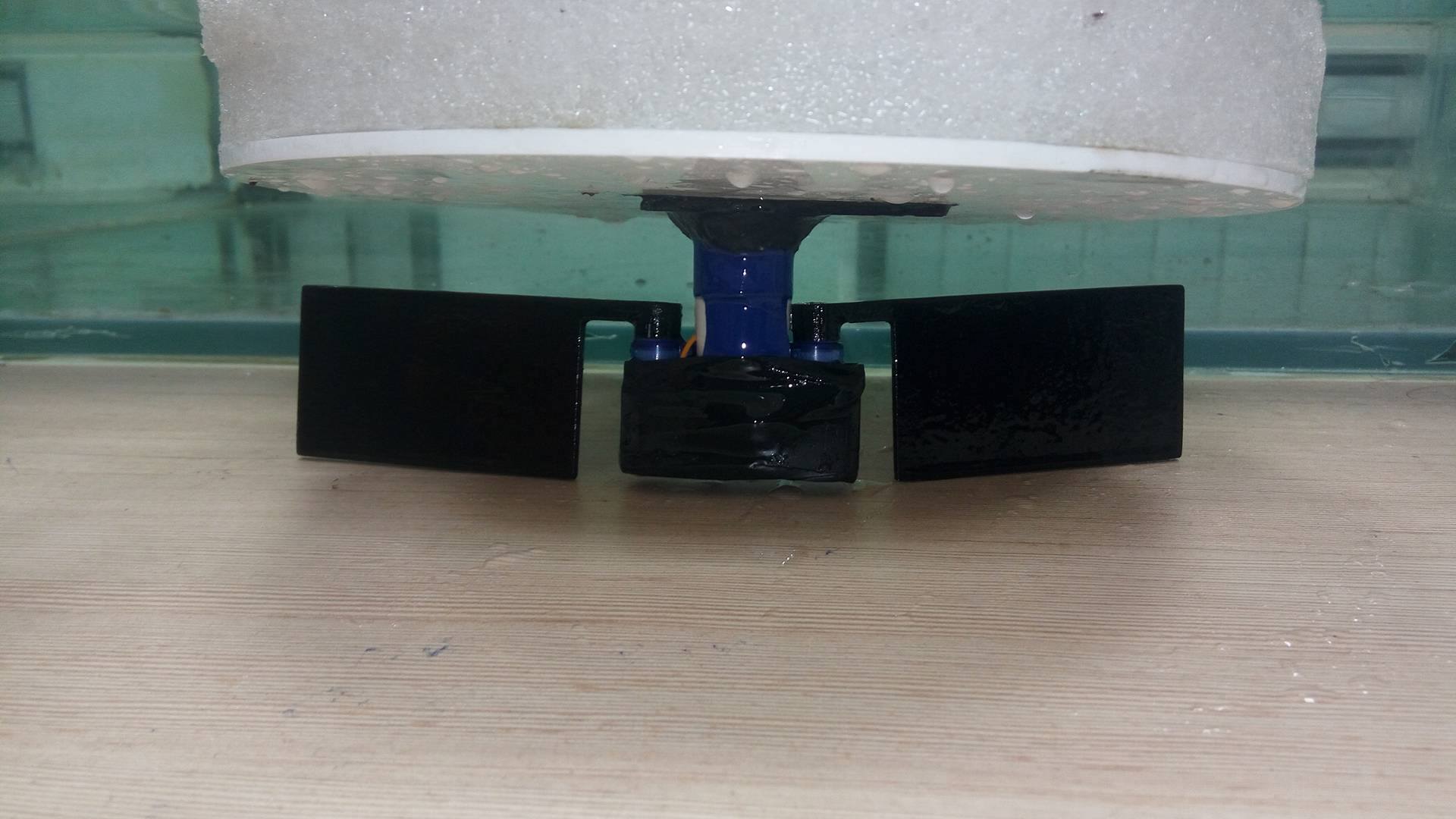}
\caption{Swimmer prototype}
\label{setup_complete}
\end{figure}

\subsection{Swimmer Construction}
The swimmer model consists of two parts: the lower part which is submerged in the fluid and the upper part which floats on the fluid. The lower part comprises of a nylon 3-D printed outer links and a casing which acts as the base link of the swimmer and contains the servo motors. The servos used are Corona CS939MG which are able to span an angle of 180 degrees. The servo dimensions are 22.5$\times$11.5$\times$24.6mm and they weigh $\approx$12.5 gram each. The casing (middle link) contains the two servos and the electrical wires are taken out to the floating part through a protrusion at the top. The dimensions of the middle link and side links are as shown in figure \ref{Dimension}.

\begin{figure}[h!]
\centering
\includegraphics[width=0.45\textwidth]{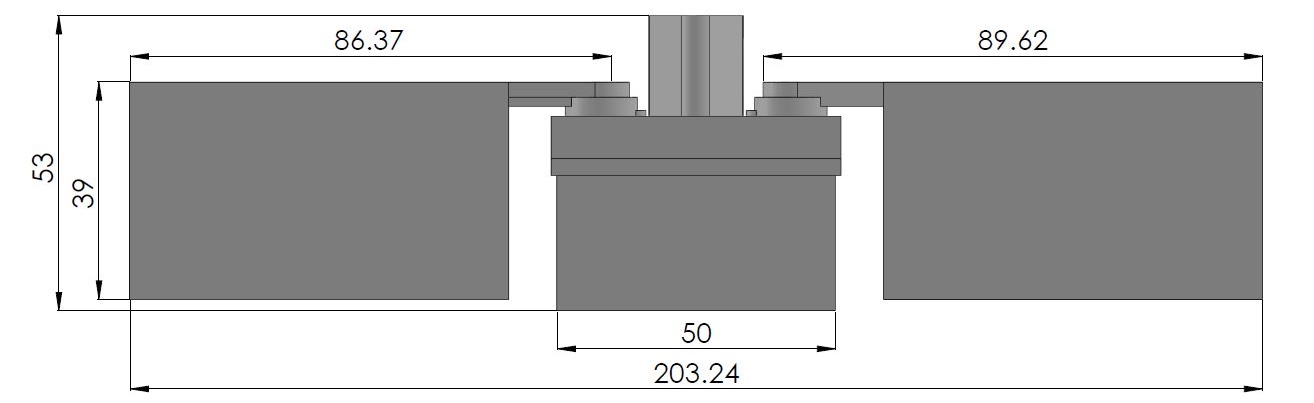}
\caption{Dimensions of the swimmer}
\label{Dimension}
\end{figure}

The side links are attached to the servos using a fixture and nuts. The top part contains the electrical components - batteries and a microcontroller. The microcontroller used is Arduino Nano because it is compact, complete and breadboard friendly. The batteries used are 3.7 V, 1000mAh Lithium Polymer batteries. The top part of the swimmer is covered with a black disk, shown in figure \ref{Styrofoam}. This disk has identifying patterns to detect the position of bot through the overhead camera system using image processing.

\subsection{Fluid Properties}
The fluid used for the experiment was glycerine, which was selected because of its high viscosity of 0.950 Pa-s, thus enabling us to achieve a Reynolds number $\approx$ 4. Approximately 0.1 cubic meter of the fluid was used in a container of size 3ft$\times$3ft$\times$1ft to allow experimenting various trajectories of the swimmer. 

\subsection{Weight neutrality and Stability}
\label{styrofoam}
To stabilize and make the swimmer buoyant, a circular Styrofoam sheet of diameter 200mm and thickness ~25mm is attached to an acrylic sheet which is fixed to the protrusion from middle link of the swimmer using an adhesive. The styrofoam sheet also prevents the model from tilting by maintaining a favorable metacentric height.
\begin{figure}[b]
\centering
\includegraphics[width=0.42\textwidth]{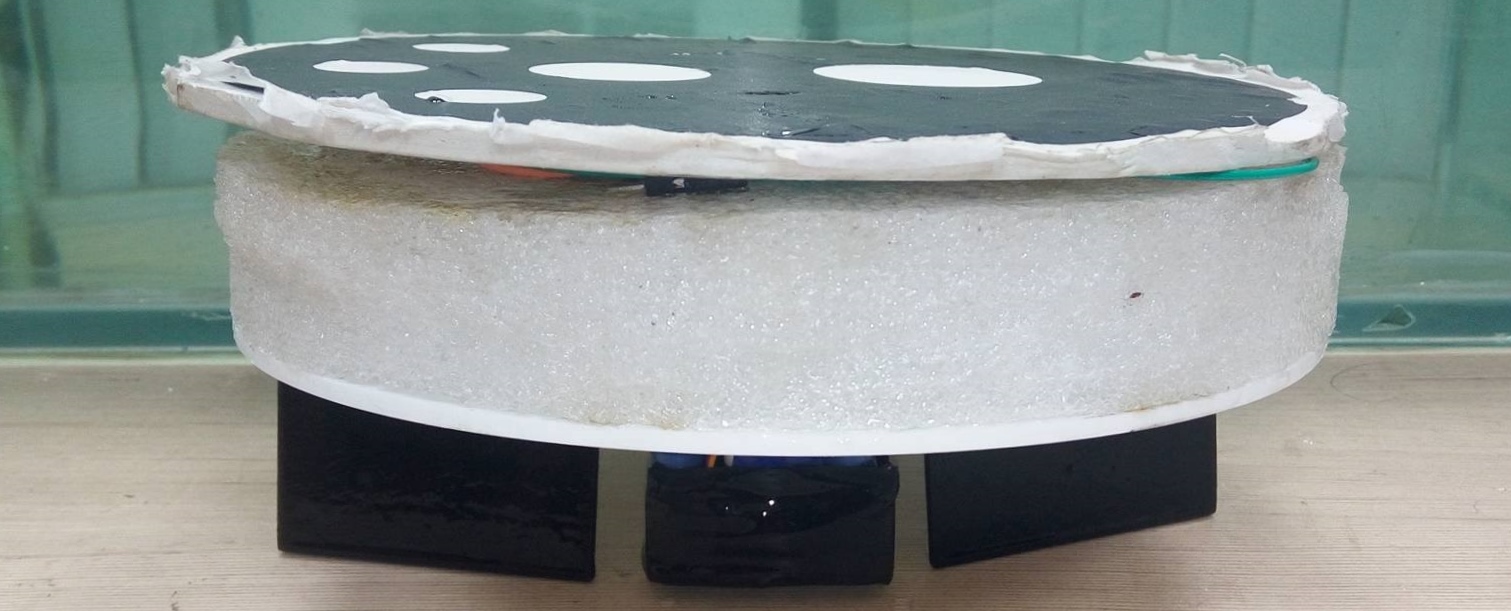}
\caption{Swimmer model with electrical system, styrofoam}
\label{Styrofoam}
\end{figure}

\subsection{Implementation of Control Law}
The control law is directly programmed on the microcontroller for an open loop analysis. The servos are controlled using a library “VarSpeedServo” \cite{68}, which helps it to move to a particular angle at a given speed. The entire control system developed using the primitives concept is hard coded in controller as a series of actuations of the side links. Since servo motors work in close loop control by default, any angle or speed feedback from servos are not necessary. The sensing for the base link's position and orientation is explained in the next section.

\subsection{Validation and vision system}
To validate the experimental and theoretical data, a vision based system is implemented. This is achieved using a ceiling mounted webcam which records the complete tank area. A continuous stream from webcam is fed to an image processing MATLAB algorithm, which identifies translational and angular position of the center of middle link using a black disk with white circles on it, relative to a predefined reference frame. This data is then compared with the theoretical values generated from MATLAB simulation which is presented in section VI.

\begin{figure}[H]
\centering
\includegraphics[width=0.3\textwidth]{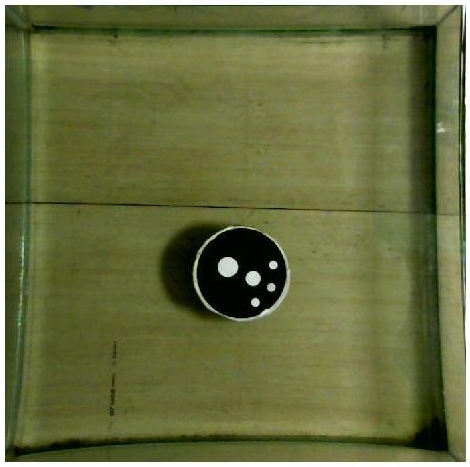}
\caption{Webcam view of the tank and swimmer}
\label{visual_system}
\end{figure}

\subsection{Swimmer Iterations}
The current prototype is the second iteration in our attempts to fabricate the Purcell Swimmer for theoretical and experimental studies. The initial model was designed to be weight neutral in glycerine on the basis of preliminary analysis. However, a few challenges were experienced while performing experiments on this version of the hardware, listed below -
\begin{enumerate}
  \item The swimmer would topple about its centroid and thus lead to rotational instability about the desired upright position. This problem was solved by using a styrofoam sheet in the new swimmer, explained in section \ref{styrofoam}.
  \item Earlier the electrical circuit was kept inside middle link thus requiring elaborate measures to waterproof the system every time the swimmer was submerged. This was solved by taking all the electric circuitry out of the swimmer, thereby making it convenient to handle. The middle link in the new prototype was completely waterproofed using spot welding on the joints.
\end{enumerate} 

\section{Simulation to get the drag coefficients}

In order to estimate the viscous drag coefficients of the links involved in the the system's model, we performed a simulation in ANSYS software where
\begin{enumerate}
  \item Flow over a horizontal plate model was used to estimate the tangential drag force, and
  \item Flow across a vertical plate model was used to estimate normal drag force.
\end{enumerate}
The input conditions, particularly the velocity of fluid was set using estimates on flow velocity in actual experiment. The fluid properties were set to match the viscosity and density of glycerine. The CAD model for links was used for simulation. The normal drag force was obtained to be 0.005922 N and tangential drag force to be 0.0001013N, which are used to obtain the viscous drag coefficients. The plots are shown in figure \ref{ansys}.

\begin{figure}[h]
\centerline{%
\includegraphics[width=0.22\textwidth]{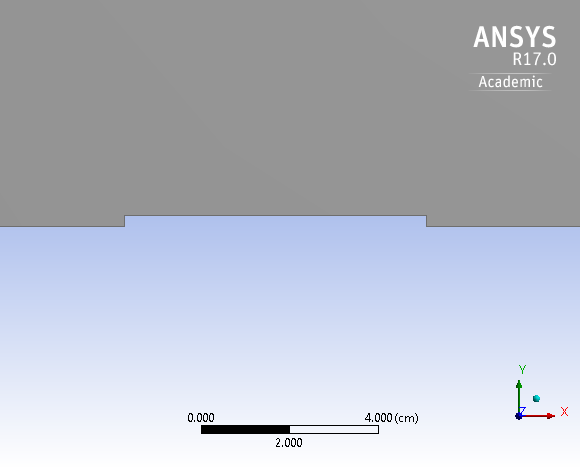}%
\includegraphics[width=0.22\textwidth]{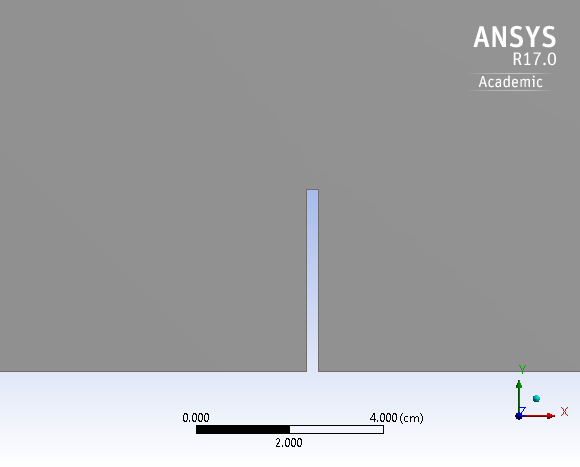}%
}%
\caption{On left, flow over horizontal  plate used for tangential force estimation and on right, flow over vertical plate used for normal force estimation}
\label{ansys}
\end{figure}

\begin{figure}[h]
\centerline{%
\includegraphics[width=0.22\textwidth]{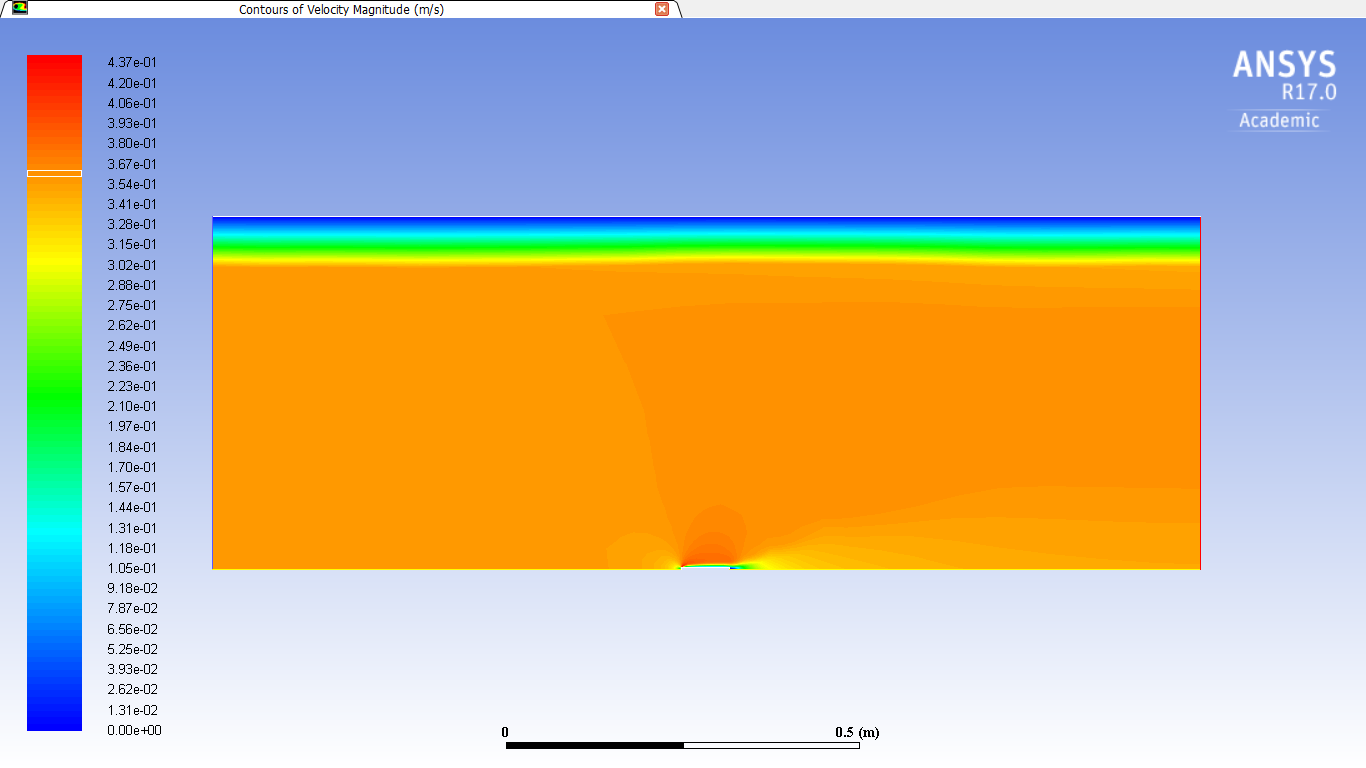}%
\includegraphics[width=0.22\textwidth]{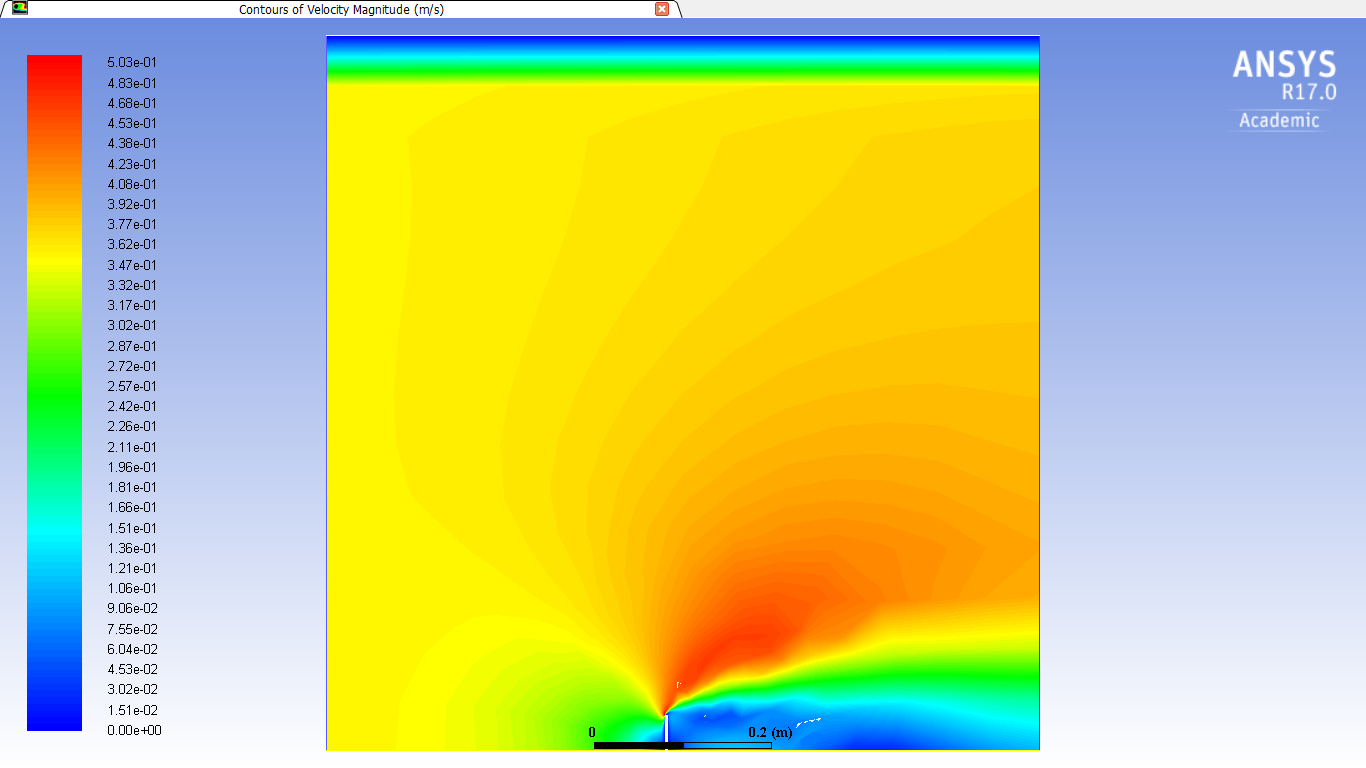}%
}%
\caption{On left, velocity profile over horizontal plate used for tangential force estimation and on right, velocity profile over vertical plate used for normal force estimation}
\label{ansys}
\end{figure}

\section{Theory on locomotion using primitives}
Since the base variables $\alpha_1$, $\alpha_2$ are fully actuated, the planar Purcell swimmer's equations in (\ref{pure_kinematic}) can be written in a driftless control-affine form as
\begin{align}
\begin{bmatrix}
\dot{\alpha}_1 \\
\dot{\alpha}_2 \\
\xi
\end{bmatrix} \nonumber = \begin{bmatrix}
1 \\
0 \\
A_1(\alpha_1,\alpha_2)
\end{bmatrix} \dot{\alpha}_1 + \begin{bmatrix}
0 \\
1 \\
A_2(\alpha_1,\alpha_2)
\end{bmatrix} \dot{\alpha}_2 \\ = g_1(\alpha_1,\alpha_2)\dot{\alpha}_1 + g_2(\alpha_1,\alpha_2) \dot{\alpha}_2 \label{pure_kinematic2} \end{align}
where, $A_1, A_2$ are the 2 columns of the local connection form $A$, and $g_1$, $g_2$ are the control vector fields. Using Chow's theorem it is shown that the Purcell's swimmer is strongly controllable at all the points in its configuration space \cite{11}. Thus, the Lie algebra generated by the control vector fields $g_1, g_2$ satisfies the following condition
\begin{align*}
& span\{g_1, g_2, [g_1, g_2], [g_1,[g_1,g_2]], [g_2,[g_1,g_2]] \} = span \{T_qQ\},
\\ & \text{where,} \:\: dim(T_qQ) = 5, \:\: \forall q \in Q
\end{align*}
A natural question that arises after coming up with controllability results is whether one can compute the controls which generate a desired trajectory in its configuration space. This section presents the work done on open loop gait design of the swimmer using the idea of motion primitives. We approach the trajectory design problem for the planar Purcell's swimmer with a simple motion planning task of point to point control, referred from \cite{66}.
\subsection{Existence of primitive controls}
We seek to control the motion of the swimmer by applying primitive controls. For $m$ dimensional control inputs, let $U = \{e_1, \cdots, e_m, -e_1, \cdots, -e_m\}$, where $e_i$'s are the unit inputs on $i$'th control. We denote by $\mathcal{U}_{prim}^m$ the collection of piecewise constant U-valued controls. Then motion planning problem using primitives is the one which takes control inputs from $\mathcal{U}_{prim}^m$. In this case, a controlled trajectory will be a composition of the integral curves of the vector fields $X_1, \cdots, X_m$. We denote by $\Phi_t^X (q)$ the integral curve or the flow of the vector field $X$ for time duration $t$ starting from $q \in Q$. We make use of the following theorem for motion planning using primitives \cite{40}.
\begin{theorem}
Let $(Q, \mathcal{V}=\{X_1, \cdots , X_m \}, \mathbb{R}^m)$ be a $C^{\infty}$ - driftless system with the vector fields $X_1, \cdots, X_m$ complete. Suppose that $Lie^{(\infty)}(V) = TQ$. If $Q$ is  connected, then, for each $q_0, q_1 \in Q$, there exist $k \in \mathbb{N}$, $t_1, \cdots, t_k \in \mathbb{R}$ and $a_1, \cdots, a_k \in \{1, \cdots, m\}$ such that 
\begin{equation}\label{primitives}
q_1 = \Phi_{t_k}^{g_{a_k}} \cdots \Phi_{t_1}^{g_{a_1}} (q_0)
\end{equation}
\end{theorem}
\vspace{10pt}
Circle $\mathbb{S}^1$ and the Special Euclidean group $SE(2)$ are connected manifolds. Also, the product of connected manifolds is also connected, hence the Purcell's swimmer's configuration space $\mathbb{S}^1 \times \mathbb{S}^1 \times SE(2)$ is a connected manifold. A vector field on $Q$ is complete if every Cauchy sequence converges in $Q$. We show in appendix 1 that the control vector fields $g_1, g_2$ are complete. Thus, using theorem 1, we conclude that there exist motion primitives for our system satisfying \ref{primitives}.

\subsection{Control synthesis}
The theorem 1 guarantees the existence of the control vector fields which generate the flow between any 2 given points on a configuration manifold $Q$ of the swimmer. Now we synthesize the primitive controls achieving certain control objectives. The flow of vector fields $X,Y$ satisfy the following rules under the addition and Lie bracket of vector fields \cite{13} -
\begin{equation}\label{flow_under_addition}
\Phi_t^{X+Y} = \lim_{n\to\infty} (\Phi_{t/n}^X \circ \Phi_{t/n}^{Y})^n
\end{equation}
\begin{equation}\label{bracket_first_order}
\Phi_t^{[X,Y]} = \lim_{n\to\infty} (\Phi_{\sqrt{\frac{t}{n}}}^{-Y} \circ \Phi_{\sqrt{\frac{t}{n}}}^{-X} \circ \Phi_{\sqrt{\frac{t}{n}}}^{Y} \circ \Phi_{\sqrt{\frac{t}{n}}}^{X})^n
\end{equation}
Also, Lie bracket under scalar multiplications follows following equation
\begin{equation}\label{bracket_scalar_mult}
c[X,Y] = [cX, Y] = [X, cY], \:\:\:  c \in \mathbb{R}
\end{equation}
From the swimmer's model driftless form of equations its seen that the $2$ base velocities are fully actuated. By the controllability result in \cite{11}, the Lie algebra upto second order Lie bracketing operation fully spans the $5$ dimensional tangent space at $p = (0,0,0,0,0) \in Q$. By elementary linear algebra we can show that the vectors $[g_1, g_2]_p, [g_1,[g_1,g_2]]_p, [g_2,[g_1,g_2]_p$ span the group velocities at $p$. In particular, there exist scalars $\alpha, \beta, \gamma \in \mathbb{R}$, such that the vector field $\alpha [g_1, g_2] + \beta [g_1, [g_1, g_2]] + \gamma [g_2, [g_1, g_2]]$, evaluated at a given point shall give any of the desired basis for the group velocities - $ \frac{\partial}{\partial x}, \frac{\partial}{\partial y}, \frac{\partial}{\partial \theta} $. Thus, using equations \ref{flow_under_addition}, \ref{bracket_first_order}, \ref{bracket_scalar_mult} we get the following expression for the flow along the linear combinations of the vector fields from the Lie algebra of control vector fields in terms of the flow along the primitives $g_1, g_2$ as
\begin{align}\label{lin_combination_brackets}
& \Phi_t^{\alpha [g_1, g_2] + \beta [g_1, [g_1, g_2]] + \gamma [g_2, [g_1, g_2]]} \\ 
& = \lim_{n\to\infty} (\Phi_{t/n}^{\alpha[g_1, g_2]} \circ \Phi_{t/n}^{\beta[g_1,[g_1, g_2]]} \circ \Phi_{t/n}^{\gamma [g_2, [g_1, g_2]]})^n \nonumber \\ 
& = \lim_{n\to\infty} [(\Phi_{\sqrt{\frac{t}{n}}}^{-g_2} \circ \Phi_{\sqrt{\frac{t}{n}}}^{-\alpha g_1} \circ \Phi_{\sqrt{\frac{t}{n}}}^{g_2} \circ \Phi_{\sqrt{\frac{t}{n}}}^{\alpha g_1}) \nonumber
\\
& \qquad \qquad \circ [(\Phi_\frac{\sqrt{\sqrt{t}}}{n}^{g_2} \circ \Phi_\frac{\sqrt{\sqrt{t}}}{n}^{g_1} \circ \Phi_\frac{\sqrt{\sqrt{t}}}{n}^{-g_2} \circ \Phi_\frac{\sqrt{\sqrt{t}}}{n}^{- g_1})^n \circ \Phi_\frac{\sqrt{t}}{n}^{-\beta g_1} \nonumber
\\ & \qquad \qquad \circ (\Phi_\frac{\sqrt{\sqrt{t}}}{n}^{-g_2} \circ \Phi_\frac{\sqrt{\sqrt{t}}}{n}^{-g_1} \circ \Phi_\frac{\sqrt{\sqrt{t}}}{n}^{g_2} \circ \Phi_\frac{\sqrt{\sqrt{t}}}{n}^{g_1})^n \circ \Phi_{\frac{\sqrt{t}}{n}}^{\beta g_1}]^n \nonumber \\
& \qquad \qquad \circ [(\Phi_\frac{\sqrt{\sqrt{t}}}{n}^{g_2} \circ \Phi_\frac{\sqrt{\sqrt{t}}}{n}^{g_1} \circ \Phi_\frac{\sqrt{\sqrt{t}}}{n}^{-g_2} \circ \Phi_\frac{\sqrt{\sqrt{t}}}{n}^{- g_1})^n \circ \Phi_{\frac{\sqrt{t}}{n}}^{- \gamma g_2} \nonumber \\
& \qquad \qquad \circ (\Phi_\frac{\sqrt{\sqrt{t}}}{n}^{- g_2} \circ \Phi_\frac{\sqrt{\sqrt{t}}}{n}^{-g_1} \circ \Phi_\frac{\sqrt{\sqrt{t}}}{n}^{g_2} \circ \Phi_\frac{\sqrt{\sqrt{t}}}{n}^{ g_1})^n \circ \Phi_{\frac{\sqrt{t}}{n}}^{\gamma g_2}]^n] \nonumber
\end{align}
Thus we have the expression for synthesizing the sequence of primitive control and the corresponding actuation-durations to achieve a velocity in a given group direction.

\section{Experimental results}
Using the equation \ref{lin_combination_brackets} which gives us the expression for control sequence of $u_1 = \dot{\alpha}_1$ and $u_2 = \dot{\alpha}_2$ we simulate the trajectory in order to generate a body velocity along one of the basis of group directions $\frac{\partial}{\partial x}, \frac{\partial}{\partial y}, \frac{\partial}{\partial \theta} $. The table below gives the values of the $\alpha, \beta, \gamma$ which decide the coefficients for linear combination in equation \ref{lin_combination_brackets} for achieving group velocity along one of the 3 directions.

\vspace{10pt}
\begin{tabular}{|l|c|c|c|}
    \hline
    \textbf{Motion (along basis)} & \textbf{$\alpha$} & \textbf{$\beta$} & \textbf{$\gamma$} \\
    \hline\hline
    Pure X ($\frac{\partial}{\partial x}$) & 0.03705 & 0 & 0\\
    \hline
    Pure Y ($\frac{\partial}{\partial y}$) & 0 & -0.0306 & 0.0306\\
    \hline
    Pure rotation ($\frac{\partial}{\partial \theta}$) & 0 & 0.0306 & 0.0306\\
    \hline
\end{tabular}
\captionof{table}{Coefficients of motion primitives} \label{tab:table1} 
\vspace{10pt}

\subsection{Motions along the group basis}
The first set of experiments was done to achieve the motion  along $\frac{\partial}{\partial x}, \frac{\partial}{\partial y}, \frac{\partial}{\partial \theta} $. The control profiles were found using \ref{lin_combination_brackets} with $t=1 sec$. Figures \ref{X motion}, \ref{Y motion} and \ref{Pure rotational motion} gives the simulation results showing net motions along $\frac{\partial}{\partial x}, \frac{\partial}{\partial y}, \frac{\partial}{\partial \theta}$ respectively.
\begin{figure}[!htb]
\centering
\includegraphics[width=1.0\linewidth]{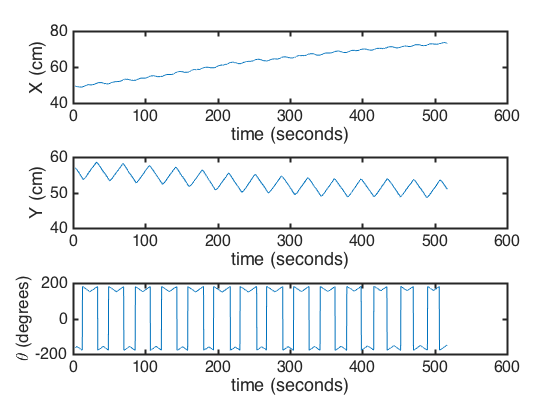}
\caption{Experimental results for net motion in X direction}
\label{X motion}
\end{figure}

\begin{figure}[!htb]
\centering
\includegraphics[width=1.0\linewidth]{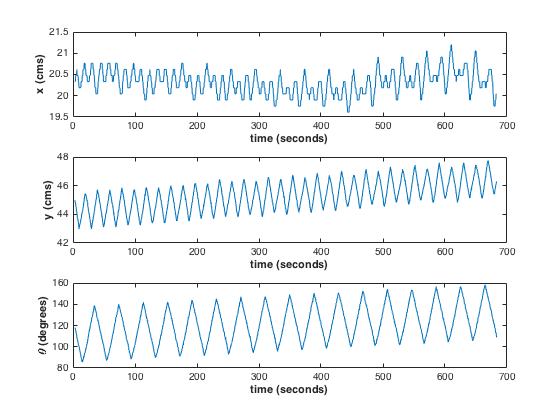}
\caption{Experimental results for net motion in Y direction}
\label{Y motion}
\end{figure}

\begin{figure}[!htb]
\centering
\includegraphics[width=1.0\linewidth]{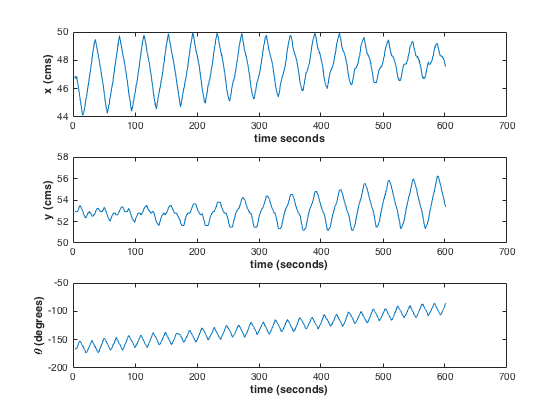}
\caption{Experimental results for net angular motion}
\label{Pure rotational motion}
\end{figure}

\subsection{Comparison of Lie brackets with existing symmetry based gaits}
In this subsection we explain the existing gaits given in \cite{62} which is based on identification of symmetries in the system model.
\subsubsection{Approximating X direction motion}
As given in table \ref{tab:table1}, for a pure motion in X direction we have $\alpha = 0.03705$, $\beta = 0$ and $\gamma = 0$. Also, even if we take any scaled versions of these values the motion will still take place in X-direction. Hence taking $\alpha=1,\quad \beta = 0 ,\quad \gamma = 0$ there are four different ways we can write gaits for $\phi^{[g_1 g_2]}_{t}$ for n =1 as shown in the figures \ref{fig:Xliebracket1}, \ref{fig:Xliebracket2}, \ref{fig:Xliebracket31}, \ref{fig:Xliebracket41}. Therefore, the motion would remain in X direction even if concatenate all the four of these loops as they all represent the same net motion. Thus, by concatenating all four of these loops we get the gait for motion in pure X direction as given in \cite{62} as given by figure \ref{fig:gutmanX}.
\begin{figure}[!htb]
\begin{subfigure}{0.2\textwidth}
  \includegraphics[width=1.0\linewidth]{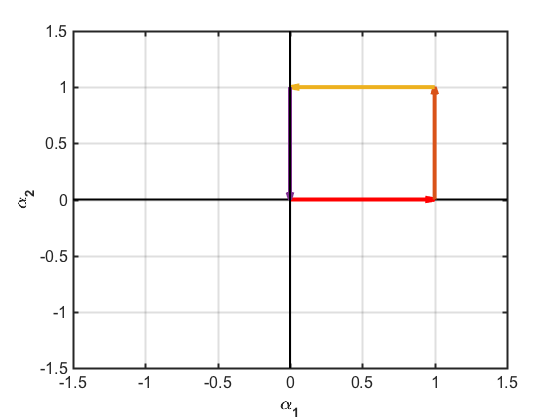}
  \caption{\scalebox{.7}{$(\phi^{-g_2}_{\sqrt{t/n}} \circ \phi^{-g_1}_{\sqrt{t/n}} \circ \phi^{g_2}_{\sqrt{t/n}} \circ \phi^{g_1}_{\sqrt{t/n}})^{n}$}}
 \label{fig:Xliebracket1}
\end{subfigure}
\hfill
\begin{subfigure}{0.2\textwidth}
  \includegraphics[width=1.0\linewidth]{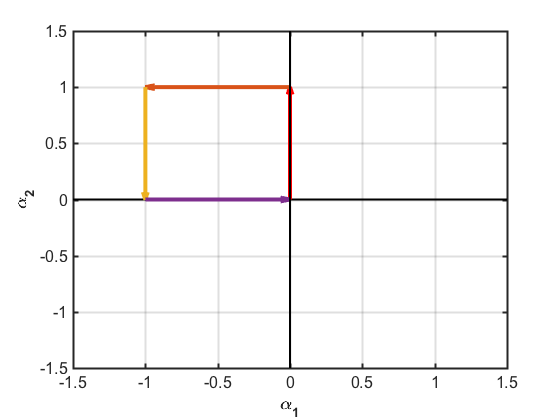}
  \caption{\scalebox{.7}{$(\phi^{g_1}_{\sqrt{t/n}} \circ \phi^{-g_2}_{\sqrt{t/n}} \circ \phi^{-g_1}_{\sqrt{t/n}} \circ \phi^{g_2}_{\sqrt{t/n}})^{n}$}}
 \label{fig:Xliebracket2} 
\end{subfigure}
\\[6pt]
\begin{subfigure}{0.2\textwidth}
  \includegraphics[width=1.0\linewidth]{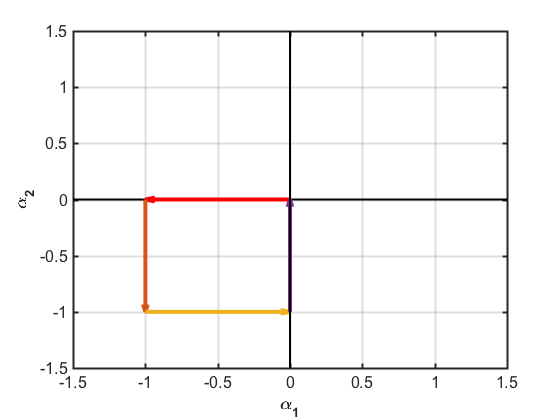}
\caption{\scalebox{.7}{$(\phi^{g_2}_{\sqrt{t/n}} \circ \phi^{g_1}_{\sqrt{t/n}} \circ \phi^{-g_2}_{\sqrt{t/n}} \circ \phi^{-g_1}_{\sqrt{t/n}})^{n}$}}
  \label{fig:Xliebracket31}
  \end{subfigure}
  \hfill
\begin{subfigure}{0.2\textwidth}
  \includegraphics[width=1.0\linewidth]{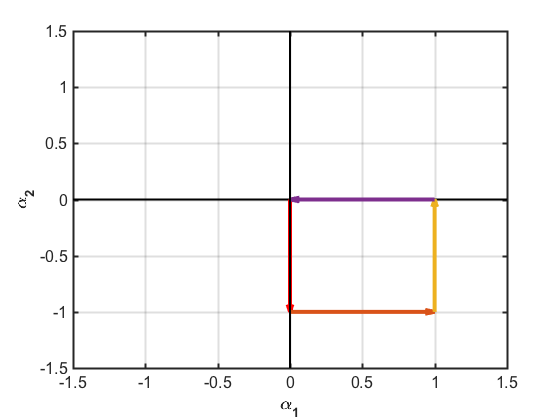}
\caption{\scalebox{.7}{$(\phi^{-g_1}_{\sqrt{t/n}} \circ \phi^{g_2}_{\sqrt{t/n}} \circ \phi^{g_1}_{\sqrt{t/n}} \circ \phi^{-g_2}_{\sqrt{t/n}})^{n}$}}
\label{fig:Xliebracket41}
\end{subfigure}
\begin{subfigure}[c]{0.4\textwidth}
\centering
\includegraphics[width=0.8\linewidth]{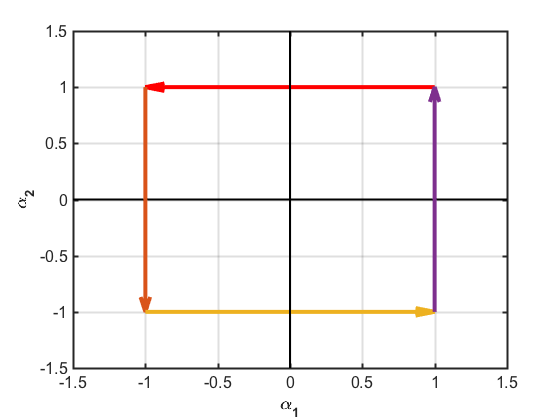}
\caption{Gait for pure motion in X direction as given in  (\cite{62})}
\label{fig:gutmanX}
\end{subfigure}
\caption{Motion in Pure X-direction}
\label{fig:Xlie}
\end{figure}
\subsubsection{Approximating Y and $\theta$ direction motion}
As given in table \ref{tab:table1}, the coefficients for motion in pure Y direction are $\alpha=0$, $\beta = -0.0306$ and $\gamma = 0.0306$. Motion primitive gait for motion in pure Y direction is shown in figure \ref{fig:Yliebracket1}for n = 1. Similarly for motion in pure $\theta$ direction are $\alpha=0$, $\beta = 0.0306$ and $\gamma = 0.0306$. Motion primitive gait for motion in pure $\theta$ direction is as shown in figure \ref{fig:thetaliebracket1}. As mentioned previously even if we take scaled versions of these coefficients the direction of motion remains the same. Thus, scaling the coefficients $\alpha$, $\beta$ and $\gamma$ such that motion corresponding to these coefficients as given by equation \ref{lin_combination_brackets} is negligible in comparison to other actuations in the system. So we can safely ignore them, and thus get gaits as shown in \ref{fig:Ymotion} and \ref{fig:thetamotion} as given in \cite{62}.
\begin{figure}[!htb]
\begin{subfigure}{0.20\textwidth}
  \includegraphics[width=1.1\linewidth]{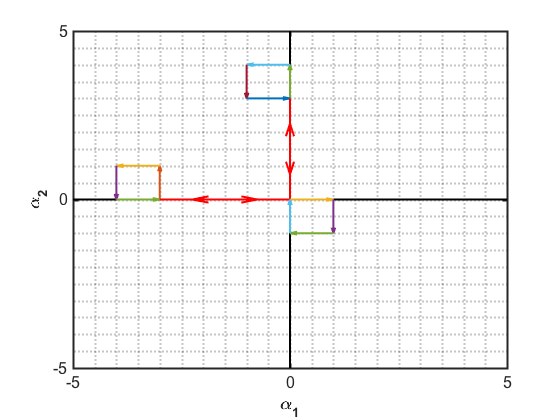}
  \caption{$\alpha =0$, $\beta=-3$ and $\gamma = 3$}
  \label{fig:Yliebracket1}
  \end{subfigure}
\hfill
\begin{subfigure}{0.20\textwidth}
\includegraphics[width=1.1\linewidth]{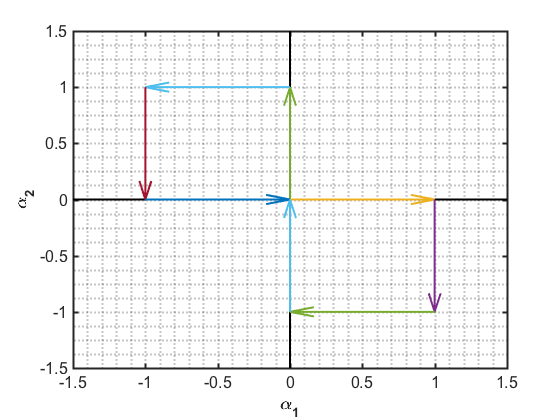}
\caption{Y motion gaits (\cite{62})}
\label{fig:Yliebracket2}
\end{subfigure}
\caption{Motion in pure Y direction(\cite{62})} 
\label{fig:Ymotion}
\end{figure}
\begin{figure}[!htb]
\begin{subfigure}{0.20\textwidth}
  \includegraphics[width=1.1\linewidth]{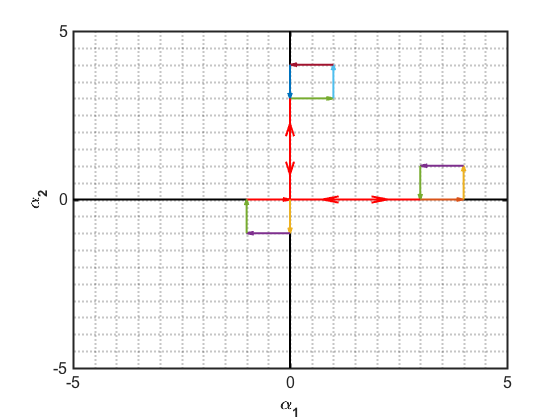}
  \caption{$\alpha =0$, $\beta=3$ and $\gamma = 3$}
  \label{fig:thetaliebracket1}
  \end{subfigure}
\hfill
\begin{subfigure}{0.20\textwidth}
\includegraphics[width=1.1\linewidth]{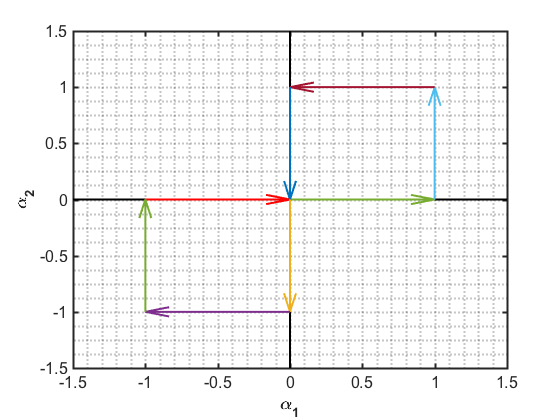}
\caption{$\theta$ motion gaits (\cite{35})}
\label{fig:thetaliebracket2}
\end{subfigure}
\caption{Motion in pure $\theta$ direction(\cite{62})} 
\label{fig:thetamotion}
\end{figure}

\subsection{Trajectory tracking results}
The primitive motions explained in the previous section have been used to track arbitrary trajectories of the base link in SE(2). We consider the following 2 cases.
\subsubsection{Tracking of an arbitrary straight line}
\begin{itemize}
\item The objective was to move along a straight line of length 12 cm in the direction 154 deg from the positive x axis of the reference frame
\item This was accomplished by first doing a $\theta$ maneuver to rotate the base link from being parallel to x axis to that at 26 deg and then translate along swimmer’s x-direction for 20 cm
\item Figure \ref{Straight line tracking} shows the experimental results of translational position of the center of the base link with respect to the reference frame. 
\item The denser part near the point $(24,62)$ in figure \ref{Straight line tracking} corresponds to angular motion to align the base link with the desired direction of motion
\item The time required for initial rotation was roughly 260 seconds and time for x direction was roughly 720 seconds
\end{itemize}

\begin{figure}[!ht]
\centering
\includegraphics[width=1.0\linewidth]{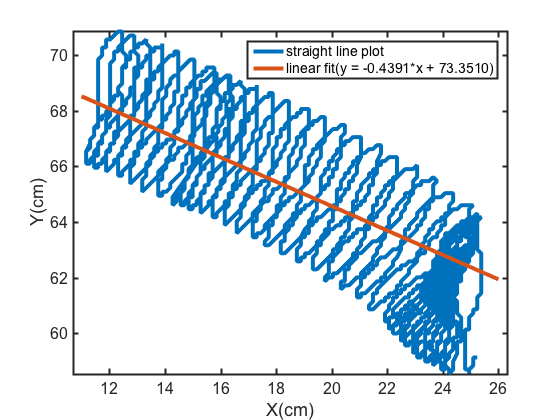}
\caption{Experimental results - tracking a straight line}
\label{Straight line tracking}
\end{figure}

\subsubsection{Tracking of circular trajectory}
\begin{itemize}
\item The objective was to track a circular trajectory of radius $20 cm$ approximated as a 10 sided polygon with length of sides as $12cm$
\item From each vertex first a theta maneuver is performed to rotate by 36 degree and then translate along the base link axis by pure x-maneuver by $12 cm$
\item The time required for each theta rotation was roughly 360 seconds and the time for x direction motion was roughly 720 seconds
\item These maneuvers were repeated 10 times in a loop to follow the complete trajectory. The tracking of the entire polygon took around 3 hours to complete
\item Figure \ref{Circle tracking} shows open loop tracking results with respect to a reference frame. Also shown is a best fitting circular trajectory depicting the error in the system. The swimmer starts at the point $(50,47)$ and follows each of the sides of the polygon.
\item For each maneuver an error is observed as can be seen in the plots of individual x, y and theta maneuvers. This is attributed to following sources -
\begin{itemize}
\item Since the experiment for tracking runs for approximately 3 hours, the servo motor tends to give a few jerks after a long usage time in fluid. This was a deviation from commanded control trajectory.
\item Although, the control inputs for a pure motion along one of $ \frac{\partial}{\partial x}, \frac{\partial}{\partial y}, \frac{\partial}{\partial \theta} $ also gave small but nonzero displacement in the other 2 directions. This is also observed from the results in figures \ref{X motion} to \ref{Pure rotational motion}.
\item Since this is an open loop control, any deviation occurring due to noises in the system get accumulated
\end{itemize}
\end{itemize}

\begin{figure}[!htb]
\centering
\includegraphics[width=1.0\linewidth]{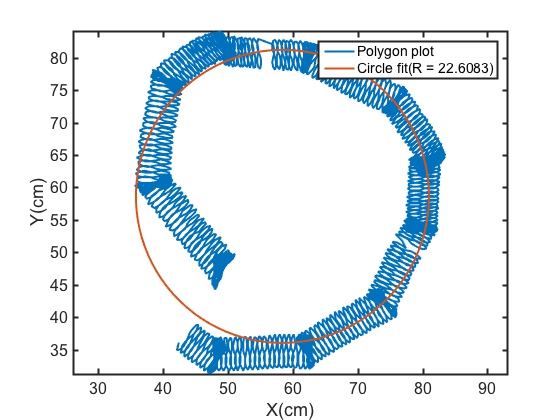}
\caption{Experimental results - tracking a circular trajectory}
\label{Circle tracking}

\end{figure}

\section{Conclusion and future work}
This paper presents a motion planning algorithm for a $3-$ link planar Purcell's swimmer. To our knowledge, this is the first work which proves the existence of motion primitives for a micro-swimming mechanism, followed by synthesis of the control inputs to track an arbitrary trajectory of the swimmer's base link in the Special Euclidean group $SE(2)$. Moreover, the approach followed is applicable to a large class of driftless control-affine systems. Hence, given the range of applications of microswimming robots, the work presented in this paper shall not only help to devise similar prototypes and motion planning algorithms, but also help in understanding mechanisms of locomotion of similar biological systems. Since most of the real life systems would work in a closed loop, as a part of future work, we plan to implement of a closed loop control system by using existing vision based feedback system.

\appendix
\subsection{Proof of completeness of control vector fields}
\begin{theorem}
The control vector fields $g_1$ and $g_2$ for the Purcell's swimmer in equation \ref{pure_kinematic2} are complete.
\end{theorem}
\begin{proof}
\begin{itemize}
\item The domain of both $g_1$ and $g_2$ is $\mathbb{S}^1 \times \mathbb{S}^1$. We observe that both $g_1$ and $g_2$ never attend zero value in their domain. Hence the support of $g_1$ and $g_2$ is the entire domain $\mathbb{S}^1 \times \mathbb{S}^1$
\item We show that $\mathbb{S}^1$ is compact
\begin{itemize}
\item The set $\{1\} \subset \mathbb{R}$ is closed, and the map $f : \mathbb{R}^2 \rightarrow \mathbb{R}$ is continuous. Therefore the circle 
\begin{equation}
\{x,y\} \in \mathbb{R}^2 : x^2 + y^2 = 1\} = f^{-1}(\{1\})
\end{equation}
is closed in $\mathbb{R}^2$.
\item $\mathbb{S}^1$ is also bounded, for example, by circle $x^2 + y^2 = 2$
\item Hence, by Heine-Borel theorem we see that $\mathbb{S}^1$ is compact
\end{itemize}
\item We use a lemma that a cartesian product of $2$ compact manifolds is compact, hence $\mathbb{S}^1 \times \mathbb{S}^1$ is also compact \cite{69}
\item Thus, the support of control vector fields $g_1, g_2$ is compact
\item A vector field is complete if its support is compact \cite{70}. Hence, $g_1, g_2$ are complete vector fields.
\end{itemize}
\end{proof}

\section*{ACKNOWLEDGMENT}
We acknowledge the help of Aseem Borkar and Prof. Arpita Sinha from Systems and Control Engineering department, IIT Bombay for their assistance in setting up the vision based tracking system.

\end{document}